\documentclass[conference]{IEEEtran}
\usepackage{times}
\usepackage{epsfig}
\usepackage{amsmath}
\usepackage{amsfonts}
\usepackage{graphicx}
\usepackage{amssymb}
\usepackage{amstext}
\usepackage{latexsym}
\usepackage{color}
\usepackage{ifthen}
\usepackage{multirow}
\usepackage{verbatim}
\usepackage{array,tabularx}

\newcommand{\be}{\begin{equation}}
\newcommand{\ee}{\end{equation}}
\newcommand{\bear}{\begin{eqnarray}}
\newcommand{\eear}{\end{eqnarray}}
\newcommand{\bears}{\begin{eqnarray*}}
\newcommand{\eears}{\end{eqnarray*}}
\newcommand{\bi}{\begin{itemize}}
\newcommand{\ei}{\end{itemize}}
\newcommand{\ben}{\begin{enumerate}}
\newcommand{\een}{\end{enumerate}}

\newtheorem{theorem}{Theorem}
\newtheorem{lemma}[theorem]{Lemma}

\newtheorem{proposition}[theorem]{Proposition}
\newtheorem{example}[theorem]{Example}
\newtheorem{definition}[theorem]{Definition}
\newtheorem{corollary}[theorem]{Corollary}

\newtheorem{construction}{Construction}

\begin{document}

\title{ Fractional Repetition Codes for Repair \\in Distributed Storage Systems}
 \author{
 \authorblockN{Salim El Rouayheb and  Kannan Ramchandran   \\}
 \authorblockA{Department of Electrical Engineering and Computer Sciences \\
 University of California, Berkeley \\
 \{salim, kannanr\}@eecs.berkeley.edu}}

\maketitle

\begin{abstract}

We introduce a new class of \emph{exact Minimum-Bandwidth Regenerating} (MBR) codes for distributed storage systems,  characterized by a low-complexity \emph{uncoded} repair process that can tolerate multiple node failures. These codes consist of the concatenation of two components: an outer MDS code followed by an inner repetition code. We refer to the inner code as a \emph{Fractional Repetition} code since it consists of splitting the data of  each node  into several packets and storing multiple replicas  of each on different nodes in the system.

Our model for repair is \emph{table-based}, and thus, differs from the random access model adopted in the literature. We present constructions of Fractional Repetition codes  based on \emph{regular graphs} and  \emph{Steiner systems} for a large set of system parameters. The resulting codes are  guaranteed to achieve the storage capacity  for random access repair.  The considered model motivates  a new definition of capacity for distributed storage systems, that we call \emph{Fractional Repetition} capacity.  We provide upper bounds on this capacity  while a precise expression remains an open problem.

\end{abstract}

\section{Introduction}

Despite being formed of unreliable nodes  having a short lifespan, \emph{distributed storage systems} (DSS) are required to store data for  long periods of time with a very high reliability. Typically, nodes in the system will unexpectedly leave  for different reasons such as     hardware failures in data centers, or  peer churning in  peer-to-peer (P2P) systems. To overcome this problem, a two-fold solution can be adopted based on \emph{redundancy} and \emph{repair} \cite{Ocean}. Classical erasure codes can be used to introduce redundancy in the system to  protect the data  from being  lost when nodes fail.  In addition, to maintain a targeted  high reliability, the system is repaired whenever a node fails  by replacing it with  a new  one. 

A distributed storage system is   formed of $n$ storage nodes and  gives  the user 
the flexibility  to recover its stored file by contacting  any $k$ out of the $n$ nodes, for some $k<n$. We call this property the \emph{MDS property} of the DSS in reference to   Maximum Distance Separable (MDS) codes.  When a node fails, the system is repaired by replacing the failed node with a new ``blank'' node. The new node contacts $d$ survivor nodes, downloads encodings of their data and stores it, possibly after processing it.  The data stored on the new node should  maintain the MDS property of the DSS. In analogy with classical codes defined by the pair  $(n,k)$,  a DSS is specified by  the triplet $(n,k,d)$, where the additional parameter $d$, referred to  as the \emph{repair degree}, accounts for the repair requirement. Fig.~\ref{fig:DSS} depicts a  $(4,2,3)$ DSS where node $v_1$ has failed and has been replaced by node $v_1'$ that contacts $d=3$ survivor   nodes in the system to download its data. 

Dimakis et al.\   introduced and studied  in \cite{DGWR07, DGWWR07,  WDR07} the design of erasure codes with efficient repair capabilities, termed \emph{regenerating codes}. The authors showed the existence of a tradeoff between  storage capacity and repair bandwidth in these systems. In this tradeoff, two regimes are of special interest, the minimum-bandwidth regime  and the minimum-storage regime. This original work focused on \emph{functional repair} where the only requirement on the data regenerated at the replacement node is to maintain the MDS property of the system.
Subsequent works focused on the design of  \emph{exact} regenerating codes    that can repair the system by reproducing an exact replica of the lost data. Rashmi et al. presented in \cite{RSKR09}  constructions of  exact minimum-bandwidth regenerating (MBR) codes  for the  case of $d=n-1$, and   for all feasible values of the repair degree $d$ in \cite{RSKKisit}.
  The existence of exact regenerating  codes for the minimum-storage regenerating (MSR) case  was demonstrated in \cite{SR10, CJM10},  and  deterministic constructions  were investigated  in  \cite{RSKR09, WD09, SR09, SRKR10}. 

In this work we are interested in the construction  of  \emph{exact} minimum-bandwidth regenerating (MBR) codes that  are characterized by a  low-complexity   repair process. In this regime,  a replacement node recovers an  exact copy of the lost data by contacting  $d$ survivor  nodes and downloading and storing one packet of data from each.  To guarantee that the constructed codes have low complexity, we require them to satisfy what we call the \emph{uncoded repair} property: a survivor node reads  the exact  amount of data he needs to send to  a replacement node and forwards it  without any processing. Our motivation is that in practical systems the read/write bandwidth of the storage  nodes is the bottleneck since it is  much smaller than the network bandwidth \cite{V09}.  Regenerating codes, such as random network codes \cite{DGWWR07},   do not satisfy the uncoded repair property  in general since they require a survivor node to read all his stored data in order  to send  a linear combination of them   to the replacement node. We show  the surprising result that even with the two apparently restrictive constraints of  exact and uncoded repair, it is possible to construct optimal regenerating codes under a table-based repair model.

 
\begin{figure}[t]
\begin{center}
\includegraphics[width=0.5\columnwidth]{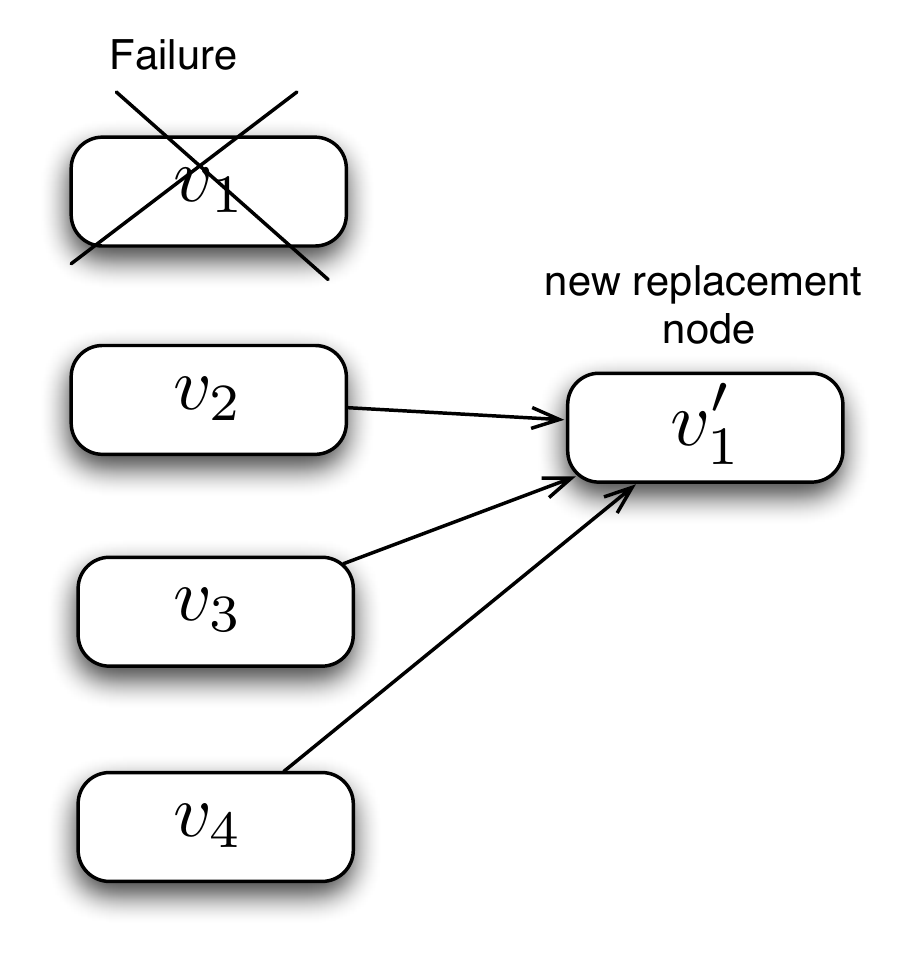}
 \caption{An example of a  distributed storage system with $(n,k,d)=(4,2,3)$.  Initially, the system is formed of $n=4$  nodes, $v_1,\dots,v_4$, storing coded packets of a file. The user contacts any $k=2$ nodes and should be able to decode the stored file. When a node fails, it is replaced by a new one that contacts $d=3$ nodes to download its data. The figure shows an instance where node $v_1$ fails and is replaced by node $v_1'$. }\label{fig:DSS}
\end{center}
\end{figure}

 Our codes are based on a generalization of the construction of  Rashmi et al.\  in \cite{RSKR09} and  are formed by  the concatenation of two component codes (see Fig.~\ref{fig:RashmiCode}): an  outer MDS  code  to  ensure the required MDS property of the DSS, and an  inner  repetition code  characterized by an  efficient uncoded repair process that is resilient to multiple node failures. The design of the inner code represents the challenging task in this construction. We refer to it as \emph{Fractional Repetition (FR)} code since, in our proposed solution,  the data stored on each node is split into $d$ packets, each of which is repeated a certain number of times in the system. We study the design of FR codes that can achieve the DSS capacity for  the MBR  point in \cite{DGWWR07} which assumes  random access repair where the new node can contact any $d$ survivor nodes.  For single failures,  we provide   a construction based on \emph{regular graphs} for all feasible values of system parameters. For the general case of multiple failures, we propose code constructions based on \emph{Steiner systems}. The table-based repair model motivates a new concept of storage capacity for distributed storage systems that we call \emph{Fractional Repetition} (FR) capacity, which we investigate.

  
The rest of the paper is organized as follows.  In Section~\ref{sec:Model}, we define our system model and motivate our code design requirements.  In Section~\ref{sec:Examples}, we give two examples of Fractional Repetition codes that are used for constructing  optimal exact MBR code with uncoded repair.  We describe code constructions based on regular graphs  for the single failure case in Section~\ref{sec:SingleFailure}.  Furthermore, we provide constructions based on Steiner systems  for the multiple failures case in Section~\ref{sec:MultipleFailures}.  In  Section~\ref{sec:FRCapacity}, we define the Fractional Repetition capacity of a DSS and provide some bounds. We conclude in Section~\ref{sec:Conclusion}  with a summary of our results and discuss related open problems.

\section{ Motivation and Model}\label{sec:Model}

%
%
 A distributed storage system DSS is defined by the triplet $(n,k,d)$, where $n$ is the total number of storage nodes in the system, $k<n$ is the number of nodes contacted by the user to retrieve his stored file, and $d\geq k$ is the repair degree that specifies the number of nodes contacted by a replacement node during repair. 
 
 We consider distributed storage systems operating in the minimum-bandwidth regime  on the storage/bandwidth  tradeoff curve described in \cite{DGWWR07}.  Our focus on this regime is motivated by the asymmetrical cost of resources in practical systems where bandwidth is  more expensive  than storage. In this case,  the repair bandwidth of the system, \emph{i.e.}, the total amount of data downloaded by a replacement node  is minimized. As a result,  the  new node   needs to download only the amount of data he will store, but no more.  
 
   For load-balancing requirements, we assume a symmetric model for  repair  where the replacement node downloads and stores an equal amount   of data, referred to as a packet, from each of the $d$ nodes it contacts. Therefore, in the minimum bandwidth regime,  $d$ also represents the node storage capacity expressed in packets.   
   
   Under this model, the storage capacity $C_{MBR}$ in packets of the DSS, representing the information-theoretic limit on the  maximum file size  that can be delivered to any user contacting $k$ out of the $n$ nodes,   was shown in \cite{DGWWR07} to be
   \begin{equation}\label{eq:Capacity}
	C_{MBR}(n,k,d)=kd-\binom{k}{2}.
\end{equation}
 
 This expression assumes a \emph{ functional repair} model where the  only constraint  on the data regenerated (stored) at the new node  is maintaining the MDS property of the system.  This allows the regenerated data  to be different from the lost data as long as it is ``functionally'' equivalent.  However, a more stringent form of repair, known as \emph{exact} repair,   that is  capable of reproducing     an exact copy of the lost data, may be required  for many system considerations such as maintaining a systematic form of the data, reducing protocol overhead and guaranteeing   data security \cite{PRR10, PRR10journal}. 
 In this case, regenerating codes are also  referred to as being exact. Recently, Rashmi et al. showed the interesting result that, in the minimum-bandwidth regime, there  is no loss in the  DSS capacity incurred by requiring exact repair, and constructed  exact  MBR codes that  achieve the capacity $C_{MBR}$ of \eqref{eq:Capacity}.

Existing code constructions suffer in general from a high complexity repair process. 
  A survivor node asked to help in repair typically has to read all his $d$ stored packets and compute a linear combination of them in order to obtain a single packet to be   forwarded to the replacement node. In addition to the computational complexity overhead, the repair process results in long delays at the survivor nodes since in general their read/write bandwidth is much smaller than the network bandwidth \cite{V09}. This motivates us to study  exact MBR codes with fast and  low-complexity  repair    where  a survivor node   reads only one of his stored  packet  and forwards it to  the replacement node with no additional processing.  We refer to this property as \emph{uncoded repair}.

  \section{ Examples}\label{sec:Examples}
  
    Before introducing our general constructions,   we  present two examples of exact regenerating  codes that can achieve the capacity $C_{MBR}$  while satisfying the    uncoded repair property.   The  first example is based on the construction of exact  MBR codes  for $d=n-1$ of Rashmi et al. in \cite{RSKR09}.
 
 %


\begin{figure}[t]
\begin{center}
\includegraphics[width=1\columnwidth]{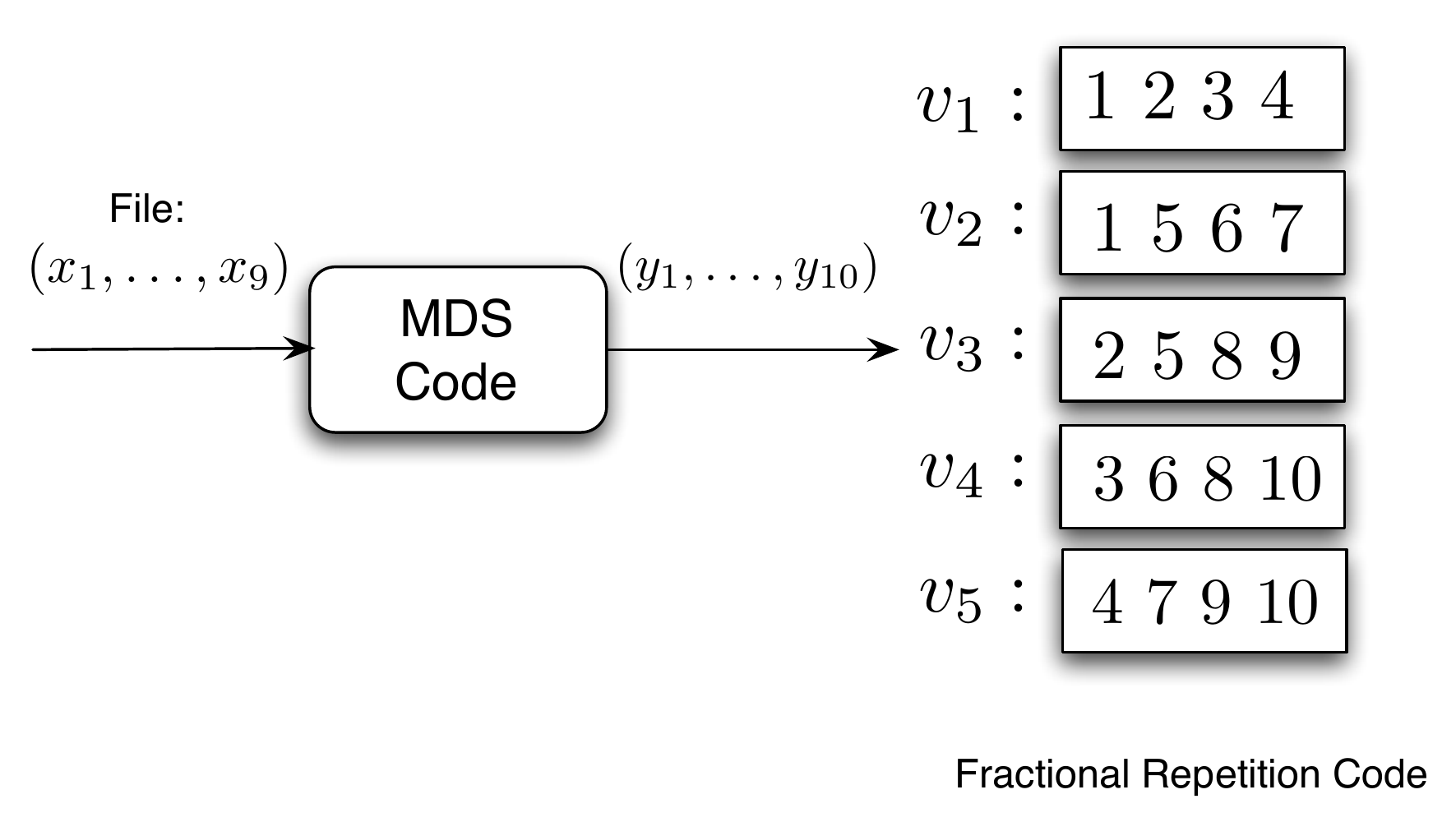}
 \caption{An exact regenerating code for a $(5,3,4)$ DSS  that can achieve the capacity $C_{MBR}=9.$ The code is  formed by a $(10,9)$ parity-check MDS code followed by a special repetition code. 
 The repetition code is defined by  listing  the indices of the coded packets  stored on each node. We refer to  the inner code as a  \emph{Fractional Repetition} code of repetition degree $\rho=2$ since the data on each node is split into $d=4$ packets  where each is repeated twice in the system. The overall code can achieve the MDS property of the DSS along with  exact and  uncoded repair in the case of one failure.}
 \label{fig:RashmiCode}
\end{center}
\end{figure}

\begin{example}\label{ex:RashmiCode}
	Consider a $(5,3,4)$ DSS where the storage capacity is equal to 9 packets according to   \eqref{eq:Capacity}. Let $X=(x_1,\dots,x_9)\in\mathbb{F}_q^9$ denote the file of 9 packets to be stored on the system, where $\mathbb{F}_q$ is the finite field of size $q$.  Figure~\ref{fig:RashmiCode}  depicts an exact MBR code that can achieve the above storage capacity  \cite{RSKR09}. This regenerating code consists  of the concatenation of  two components: an   outer  $(10,9)$ parity-check MDS code  followed by a special repetition code. The  MDS code takes the file $X$ as input and outputs the codeword $Y=(y_1,\dots,y_{10})$, where $y_i=x_i, i=1,\dots,9$, and $y_{10}$ is a parity-check packet, \emph{i.e.},  $y_{10}=\sum_{i=1}^9 x_i.$  The coded packets $y_1,\dots, y_{10}$ are then placed on the 5 storage nodes following the pattern of the inner code in Fig.~\ref{fig:RashmiCode}. That is, nodes $v_1,\dots,v_5$ store, respectively,  $\{y_1,y_2,y_3,y_4\}, \{y_1,y_5,y_6,y_7\}, \{y_2,y_5,y_8,$ $y_9\}, \{y_3,y_6,y_8,y_{10}\}$ and $\{y_4,y_7,y_9,y_{10}\}$.  
	
	A user contacting a node can download all  its stored packets.  Since the repetition code is such that any two nodes share exactly one packet, a user contacting $k=3$ nodes will be able to download $9$  distinct packets out of the $10$ coded ones (12 in total, of which 3 are repeated twice). Thus, due to the MDS property of the outer code, it can recover the whole  file $X$.  Moreover, each of the coded packets is replicated twice in the system on two different nodes  which   guarantees uncoded exact repair  in the case of a single node failure.  Indeed, whenever a  node fails, its data can be recovered exactly by contacting the four surviving nodes and downloading  one  packet from each. For instance,  when node $v_1$ fails, a replacement node contacts  nodes $v_2,\dots,v_5$, and  downloads  packets $y_1,\dots,y_4$ from each, respectively. 
\end{example}

%

The previous code is a special example of the  exact MBR codes devised in \cite{RSKR09} where uncoded repair was not a requirement. This construction, however, is limited to systems with repair degree $d=n-1$ that require contacting all the survivor  nodes when a failure occurs. This may not always be feasible  due, for example, to nodes having  limited access bandwidth or  being temporarily down. Moreover, the uncoded repair process here cannot tolerate multiple failures together, which may not be a rare event in large-scale systems where failures can also be correlated, or systems where repair is not immediate but  performed at prescheduled intervals. For these reasons,  we are interested in  exact MBR codes for systems with small repair degree $d$ that have  an uncoded repair process that  can tolerate multiple failures. 

To introduce our constructions which will be detailed in the following sections, we give a second example  for a  DSS  with $(n,k,d)=(7,3,3)$ having an exact and uncoded repair  process that can tolerate up to two failures. 

\begin{example}\label{ex:Fano}
 Consider a $(7,3,3)$ DSS where the system storage  capacity is  $C_{MBR}=6$ by \eqref{eq:Capacity}. The code that we propose for this system   is also constructed by  concatenating two constituent codes: an  outer $(7,6)$ MDS code that outputs coded packets indexed from $1$ to $7$,  followed by the  repetition code   depicted in Fig.~\ref{fig:FanoCode}(a). It can be seen that each of  the 7 packets  forming the output of the outer code is replicated on $3$ different nodes in the system. Therefore,  the system will  always  have a surviving copy of each packet  in the case of two failures. Thus,  the code guarantees exact uncoded repair  for up to  two node failures.   
   
Next, we check that this code indeed  achieves the capacity $C_{MBR}$. The structure of the inner repetition code is deduced from   the projective plane   of order $2$, also called the Fano plane,  depicted in Fig~\ref{fig:FanoCode}(b) \cite[Chap~2]{Mcwilliams}.  The Fano plane  consists of 7 points indexed from $1$ to $7$  and the following $7$ lines $123, 345, 156, 147, 257, 367$ and $246$, including the circle. To form the inner repetition code, we associate to  each line in the Fano plane  a  distinct storage node in the DSS. Then,  the three points belonging to that line give the indices of the packets stored on the  corresponding node. In the projective plane, any two lines intersect in exactly one point. Therefore, any two nodes will have exactly one packet in common.  This implies that any user that contacts 3 different nodes can get \emph{at least} $3\times3-\binom{3}{2}=6$ distinct packets. For instance, a user contacting  nodes $v_2,v_4$ and $v_5$ will get exactly 6 different packets, namely, all the packets except the one with index $6$. Whereas another user contacting $v_1,v_3$ and $v_4$ will get  all the 7 packets. In a worst-case analysis,    the capacity of the system is  limited by the user that gets the least number of packets, which is $6$ here.  Hence, the  outer MDS code allows any user to recover the stored file of 6 packets which is exactly the capacity $C_{MBR}(7,3,3)$ of \eqref{eq:Capacity}. 
\end{example}

\begin{center}
\begin{table*}[t]
\begin{tabular}{|l|l|}
\hline
Original repair model in  \cite{DGWWR07}& Repair model of FR codes\\ \hline\hline
\emph{Functional}: regenerated  data should satisfy the  MDS property.  & \emph{Exact}: regenerated data is an exact copy of the lost one.\\ \hline
\emph{Coded}: new node downloads linear combination of packets.   & \emph{Uncoded}: new node downloads  a specific packet with no coding. \\ \hline
\emph{Random-Access}: the new node contacts any $d$ surviving nodes. & \emph{Repair Table}: a table specifies the set of $d$ nodes to be contacted for repair.  \\ 
\hline
\end{tabular}
 \caption{ A comparison between the model for repair in the original work of Dimakis et al. in \cite{DGWWR07} and the model for repair for the  Fractional Repetition codes proposed here. }
   \label{tab:comparison}
\end{table*}
\end{center}
 
The previous two examples highlight the central role of the inner repetition code  that allows us to obtain  the desired uncoded and exact repair properties of the code in addition to achieving the capacity $C_{MBR}$ by carefully  placing the different copies of the coded packets on different nodes in the system. 
We call  the inner code a \emph{Fractional Repetition} (FR) code of repetition degree $\rho$  since  the content of each node is split into  $d$ packets and $\rho$ replicas of each are stored on different nodes in the system. For instance, the inner code in the first example was an FR code with $\rho=2$, whereas in the second example $\rho=3$.


Notice that in Example~\ref{ex:Fano},  a replacement node has to contact a \emph{specific} set of  $d$ nodes  for repair, depending on which nodes have failed.  For example, when node $v_1$ fails, a replacement node can recover all  the lost packets by contacting  nodes $\{v_4,v_5,v_6\}$, but not $\{v_2,v_3,v_4\}$. We assume  that there is a \emph{repair table} maintained in the system that is available to all the nodes in the DSS.  The repair table indicates for each possible failure pattern the set of nodes that can be contacted for repair, and which packet to download from each. This  table-based repair model   will be adopted throughout this paper and  differs from the random access model adopted in the literature  where repair can be performed  by contacting \emph{any} $d$ survivor node.  We believe that this relaxation in the repair model is a well-justified price to pay in order to obtain low-complexity regenerating codes, and goes along with practical system implementations that always  include a tracker server that stores the system metadata. 

Table~\ref{tab:comparison} summarizes  the differences between the repair model adopted here and the original model of  \cite{DGWWR07}.

\begin{figure}[t]
\begin{center}
\includegraphics[width=1\columnwidth]{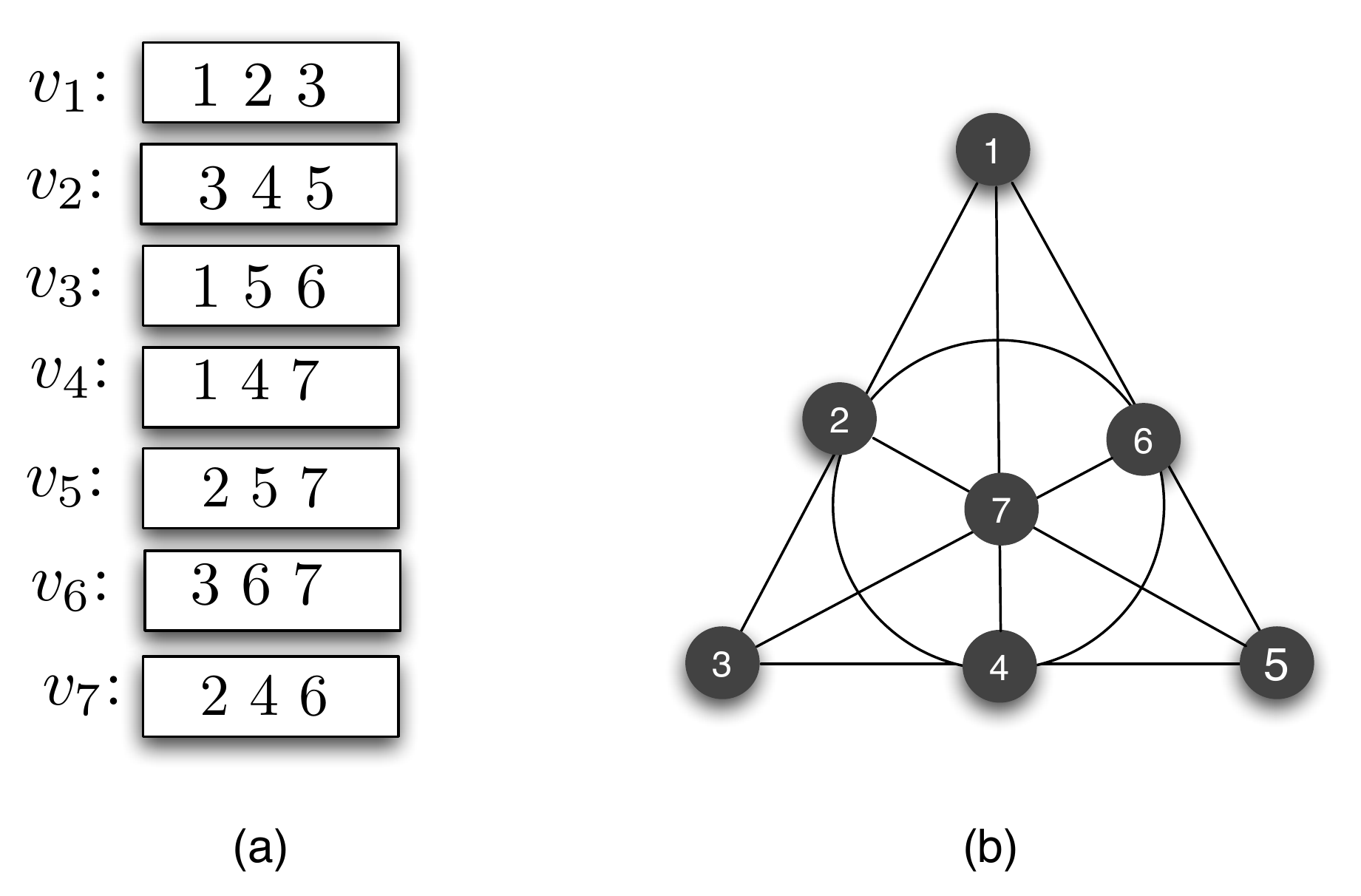}
 \caption{ (a) The inner repetition code of Example~\ref{ex:Fano} for a $(7,3,3)$ DSS.  This corresponds to a  Fractional Repetition code with repetition degree $\rho=3$.   The code structure is derived from the Fano plane depicted on the right. Each  of the 7 lines  in the Fano plane,   including the circle, corresponds to  a distinct storage node. The points lying on that line give the indices of  the packets stored on the node. The overall  code can achieve the capacity $C_{MBR}=6$ for this DSS, and has an exact and uncoded repair process. (b) The Projective plane of order 2, also known  as the Fano plane. }\label{fig:FanoCode}
\end{center}
\end{figure}


 

\section{Fractional Repetition Codes with $\rho=2$}\label{sec:SingleFailure}

The previous two examples suggest a general method for constructing exact MBR codes with  uncoded repair  process that is resilient to multiple failures. This construction consists of   concatenating  an outer MDS code with an  inner Fractional Repetition code with repetition degree $\rho$ that can tolerate up to $\rho-1 $ nodes failing together. Since MDS codes exist for all feasible parameters provided that the packets are taken from an alphabet of  large enough size, the  challenging part of the suggested construction   is    designing the Fractional Repetition code.  Assuming all the packets in the system are to be equally protected, we are motivated to provide the following general definition of FR codes:

\begin{definition}[Fractional Repetition Codes]\label{def:FR}
	A \emph{Fractional Repetition}  (FR) code $\mathcal{C}$, with repetition degree $\rho$, for an  $(n,k,d)$ DSS, is a collection $\mathcal{C}$ of $n$ subsets $V_1,V_2,\dots,V_n$ of a set $\Omega=\{1,\dots, \theta\}$ and of cardinality $d$ each, satisfying the condition  that each element of $\Omega$ belongs to exactly $\rho$ sets in the collection.
	\end{definition}

In this definition, each set $V_i$ contains the indices  of the coded packets at the output of the outer MDS code that are stored on node $v_i, i=1,\dots,n$. The value of  $\theta$, which will be determined later, corresponds to   the length of the codewords of the  outer MDS code.  For instance, following this definition, the FR code  of Example~\ref{ex:Fano} can be written as $\mathcal{C}=\{V_1,\dots,V_7\}$ with  $V_1=\{1,2,3\}, V_2=\{3,4,5\}, V_3=\{1,5,6\}, V_4=\{1,4,7\}$ $V_5=\{2,5,7\}, V_6=\{3,6,7\}, V_7=\{2,4,6\}$, where $\theta=7$ and $\Omega=\{1,\dots,7\}$.

We focus first on the design of   Fractional Repetition codes of repetition degree $\rho=2$ with an uncoded repair that is  tolerant to  a single  failure. We provide a code construction based on \emph{regular graphs} that can achieve the capacity $C_{MBR}$ of \eqref{eq:Capacity} for all feasible values of $n$ and $d$. 

To that end, 
we define the rate $R_\mathcal{C}(k)$ of an FR code $\mathcal{C}$ as the maximum file size, i.e., the maximum number of distinct packets, that the code is guaranteed to  deliver to \emph{any} user contacting $k$ nodes.

\begin{definition}[FR Code Rate]
The rate $R_\mathcal{C}(k)$ of an FR code $\mathcal{C}=\{V_1,V_2,\dots,V_n\}$  for a DSS with parameters $(n,k,d)$ is defined as 
\begin{equation}
R_\mathcal{C}(k):=\min_{\substack{
   I\subset [n]\\
   |I|=k
  }}|\cup_{i\in I}V_i|,
\end{equation}
with $[n]=\{1,\dots,n\}.$
\end{definition}

As it can be seen from the previous examples and the above definition, the DSS parameter $k$ specifying the number of nodes contacted by a user, is not intrinsically related to the construction of the FR code. An FR code designed for a DSS with parameters $(n,k_1,d)$ can be seamlessly used for another DSS with parameters $(n,k_2,d)$, with $k_1\neq k_2$. 
An FR code $\cal{C}$ is said to be   \emph{universally good}  if its rate is guaranteed to be no less than the capacity $C_{MBR}$ of the  DSS, \emph{i.e.},
$R_\mathcal{C}(k)\geq C_{MBR}(n,k,d)$, for all $k=1,\dots,d$. Here,  the inequality follows from  the fact that FR codes can have rates that exceed $C_{MBR}$ due to the table-based repair relaxation, a property that will be investigated further in Section~\ref{sec:FRCapacity}.

An $(n,k,d)$ DSS stores $nd$ packets in total. When an FR code of degree $\rho$ is used,  $\theta$ distinct packets  are stored in the system, where each is replicated exactly $\rho$ times. Therefore, the following relation exists between the FR code parameters: 
\begin{proposition}\label{prop:theta}
The parameter  $\theta$ in Def.~\ref{def:FR} of an FR code of degree $\rho$ for an $(n,k,d)$ DSS is given by,
\begin{equation}\label{eq:theta}
\theta \rho=nd.
\end{equation}
\end{proposition}


The Exact MBR  codes of Rashmi et al.\ were proposed in \cite{RSKR09} as capacity achieving codes   for the special case of $d=n-1$. In this case, when a node fails, all the remaining nodes in the system are contacted by the replacement node, which implies that the   random access and table-based repair models are equivalent.

 These codes can be viewed as special FR codes with repetition degree $\rho=2$ as shown in Example~\ref{ex:RashmiCode}.  Their general construction can be described  with the assistance of a complete graph $K_n$  defined on $n$ vertices $u_1,\dots, u_n$, with edges  indexed from $1$ to $\binom{n}{2}$.  Prop.~\ref{prop:theta} gives $\theta=\frac{n(n-1)}{2}=\binom{n}{2}$ distinct packets. The FR code is obtained by storing on node $v_i, i=1,\dots,n$, the packets having the same indices as  the edges adjacent to vertex $u_i$ in $K_n$. Figure~\ref{fig:Complete} depicts the complete graph $K_5$ with its edges indexed in a way to give the FR code of Fig.~\ref{fig:RashmiCode}.

\begin{figure}[t]
\begin{center}
\includegraphics[width=.7\columnwidth]{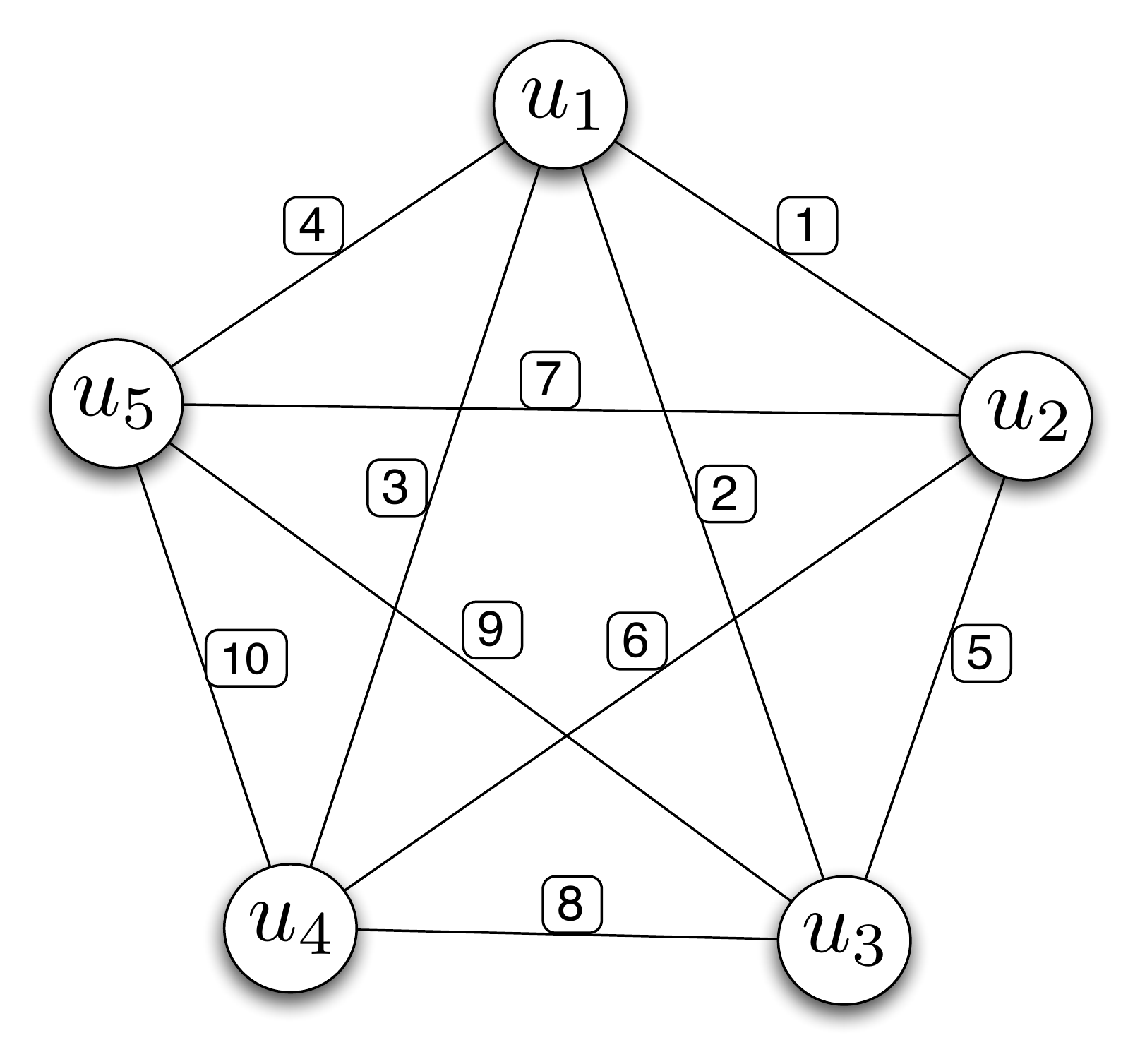}
 \caption{The complete graph $K_5$ on 5 vertices. The labeling of the edges   from $1$ to $\binom{5}{2}=10$ gives the FR code with $\rho=2$ for the DSS $(5,4,3)$ depicted in Fig.~\ref{fig:RashmiCode}. The edges adjacent to vertex $u_i$ give the indices of the packets stored on node $v_i$ in the DSS.}\label{fig:Complete}\end{center}
\end{figure}

Next, we describe a construction of FR codes with  repetition degree   $\rho=2$ and $d<n-1$. For $\rho=2$,   Prop.~\ref{prop:theta} gives a necessary condition for the existence of FR codes, that is, $nd$ should be even. We will  show that this is also a sufficient condition and provide a general code construction based on regular graphs.

 A  $d$-regular  graph $R_{n,d}$  on $n$ vertices is a \emph{simple} graph where all  vertices have the same  degree $d$, \emph{i.e.}, the same number of neighboring nodes. The graph $R_{n,d}$ has $\frac{nd}{2}$ vertices, and exists whenever $nd$ is even \cite{W99}.

\begin{construction}\label{con:regular}
An FR code with repetition degree $\rho=2$ can be constructed for an $(n,k,d)$ DSS, with $nd$ even, in the following way:
\begin{enumerate}
\item Generate a $d$-regular graph $R_{n,d}$ on $n$ vertices $u_1,\dots,u_n$.
\item  Index the edges of $R_{n,d}$ from $1$ to $\frac{nd}{2}$.
\item  Store on  node $v_i$ in the DSS the packets indexed by the edges  that are adjacent to vertex $u_i$ in the graph $R_{n,d}$.
\end{enumerate}
\end{construction}

\begin{figure}[t]
\begin{center}
\includegraphics[width=1\columnwidth]{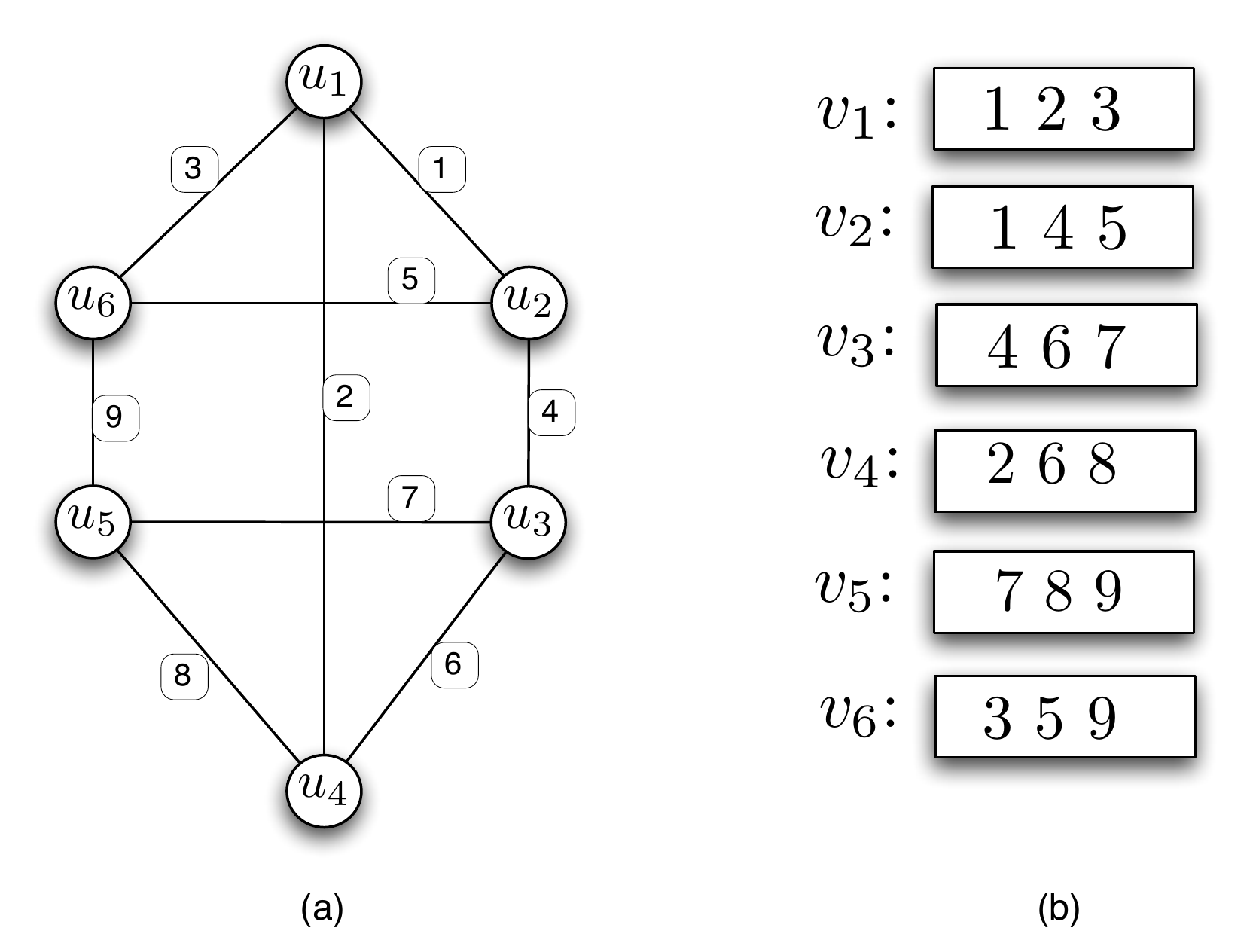}
\caption{(a) $R_{6,3}$ a $3$-regular graph on $6$ vertices. All the six vertices have constant degree equal to 3. The edges of the graph are  indexed from $1$ to $\frac{nd}{2}=\frac{6\times 3}{2}=9$. (b) The corresponding universally good FR code  with $\rho=2$ obtained by Construction~\ref{con:regular} for a DSS with $n=6$ an $d=3$.}\label{fig:Regular}
\end{center}
\end{figure}

The regular graph in Step 1 can be randomly generated using efficient randomized algorithms  that are well-studied  in the  literature, see for example \cite{KV03}.  The fact that the FR codes obtained by this construction have repetition degree $\rho=2$ is a direct consequence of  the graph being simple with each edge being adjacent to exactly two vertices. This also  implies that any two nodes cannot have in common more than one packet. Therefore, among any $k$ nodes observed by a user, there are at most $\binom{k}{2}$  repeated packets which corresponds to the case when any two nodes share a distinct packet. Therefore, we have the following lemma.

\begin{lemma}
The FR codes with repetition degree $\rho=2$ obtained by Construction~\ref{con:regular} are \emph{universally good} codes.
\end{lemma}

  Fig.~\ref{fig:Regular} shows  a 3-regular graph $R_{6,3}$ and the corresponding  universally good FR code obtained by   Construction~\ref{con:regular}  for the  DSS with $n=6$ and $d=3$.

%
%

\section{Fractional Repetition Codes with $\rho>2$}\label{sec:MultipleFailures}

Practical systems  require the repetition degree to be  at least 3 \cite{GFS},  and the previous construction based on regular graphs cannot be generalized to this case. We present here two new constructions of FR codes with $\rho>2$ based on a  combinatorial structure known as \emph{Steiner system} that can be thought of as a generalization of the projective plane of Example~\ref{ex:Fano}.


\begin{figure*}[t]
\begin{center}
\includegraphics[width=1.2\columnwidth]{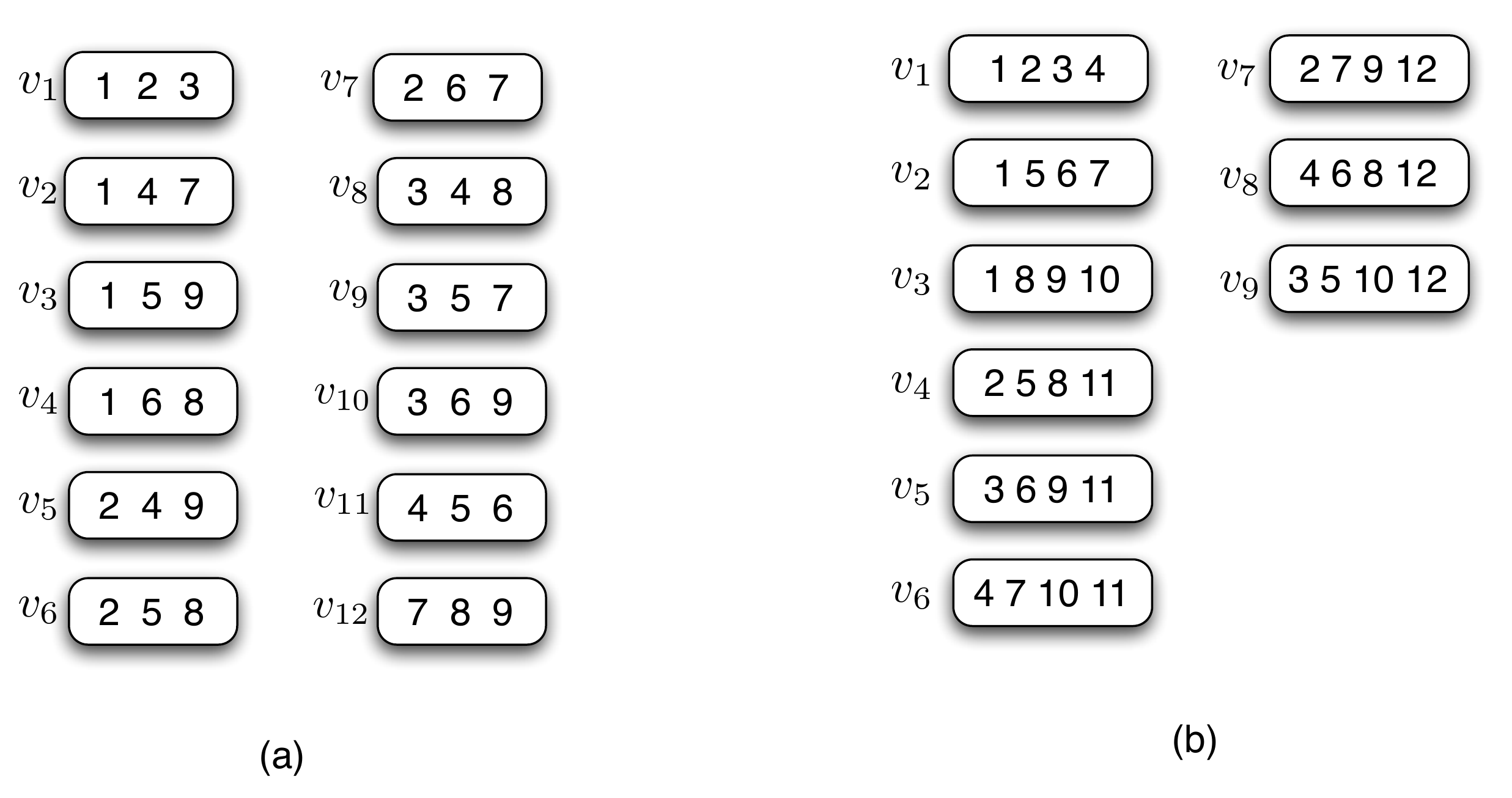}
 \caption{(a) FR code with $\rho=4$ for a DSS with $n=12$ and $d=3$ derived from the Steiner system S(2,3,9) using Construction~\ref{cons:direct}. (b) FR code with $\rho=3$ for a DSS with $n=9$ and $d=4$ derived from the same Steiner system  using Construction~\ref{cons:transpose}.}\label{fig:Transpose}
 \end{center}
\end{figure*}

\subsection{Steiner Systems}
We start by giving a definition of   a Steiner system.

\begin{definition}[Steiner System]\label{def:Steiner}
A Steiner system $S(t,\alpha,v)$ is a collection of  subsets, called blocks, $B_1,\dots,B_b$,  of size $\alpha$ of a set $\cal{V}$ containing $v$ elements, called points, with the property that any subset of  $t$ points is contained in \emph{exactly} one block.
\end{definition}

It can be shown that in a  Steiner system every point belongs to exactly the same  number  of   blocks  denoted $r$ \cite[p. 60]{Mcwilliams}. 
One way to guarantee the achievability of the capacity in \eqref{eq:Capacity} is to require that nodes do not share more than one packet. For this reason, we will be mostly interested in Steiner systems with $t=2$. Simple counting arguments give the following two well-known properties of a Steiner system $S(2,\alpha,v)$ \cite[Chap.~2]{Mcwilliams}.

\begin{proposition}\label{prop:Steiner}
The parameters $b$ and $r$ of a Steiner system $S(2,\alpha,v)$ are given by:
\begin{align}
b\alpha&=vr,\label{eq:st1}\\
v-1&=r(\alpha-1)\label{eq:st2}.
\end{align}
\end{proposition}

Equation~\eqref{eq:st1} is similar to \eqref{eq:theta} for FR codes. The Fano plane of Fig.~\ref{fig:FanoCode}(a) is an example of a Steiner system where $\cal{V}$ is the set of 7 points and the blocks are the lines (including the circle). The Fano plane is indeed $S(2,3,7)$ since there is a single line that goes through any two points.  Prop.~\ref{prop:Steiner} gives $r=3$, \emph{i.e.}, each point belongs to exactly 3 lines, and $b=7$, \emph{i.e.}, the Fano plane contains 7 lines, which can be easily checked on the figure.

For a Steiner system $S(t,\alpha,v)$ to exist, it is necessary that the parameters $b$ and $r$   given in Prop.~\ref{prop:Steiner} be integers. Wilson proved  in \cite{Wilson} that this condition is also sufficient when $v$ is large enough.

\begin{theorem}\label{th:Wilson}
Given a positive integer $\alpha$,  Steiner systems $S(2,\alpha, v)$ exist for all sufficiently large integers $v$ for which the congruences 
$vr\equiv 0 \bmod{\alpha}$ and 
$v-1\equiv 0 \bmod{\alpha-1},$
are valid. 
\end{theorem}

\subsection{Code Constructions}
We present now two constructions of universally good FR codes derived from Steiner systems. Example~\ref{ex:Fano} suggests the following direct construction.

\begin{construction}\label{cons:direct}
Given a Steiner system $S(2,\alpha,v)$  with blocks $B_1,\dots,B_b\subset \mathcal{V}=[v]$, an FR code $\mathcal{C}$  can be obtained by taking  $\mathcal{C}=\{B_1,\dots,B_b\}$. This gives an FR code with $\rho=\frac{v-1}{\alpha-1}$ and $\theta=v$ for a DSS with parameters   $n=\frac{v(v-1)}{\alpha(\alpha-1)}$ and $d=\alpha$   as given by Prop.~\ref{prop:Steiner}.\end{construction}


By definition, any two  blocks in $S(t,\alpha,v)$ cannot intersect in more than $t-1$ elements. This implies that in  the FR codes obtained by Construction~\ref{cons:direct}, two nodes can have at most one packet in common. Thus,  in  any collection of $k$ nodes,  there are at most $\binom{k}{2}$ repeated packets. Therefore,   the obtained FR codes can achieve the capacity $C_{MBR}$ for all  $k=1,\dots,d$ as stated in the following lemma.

\begin{lemma}
The FR codes obtained by Construction~\ref{cons:direct} are universally good.
\end{lemma}

Construction~\ref{cons:direct}  has the disadvantage that the two important code design parameters, $n$ and $\rho$ do not figure explicitly in the Steiner system parameters. Next, we present a second construction where $n$ and $\rho$ directly determine the Steiner system and where the repair degree $d$ is a fraction of the survivor nodes. 


\begin{construction}\label{cons:transpose}
Given a Steiner system $S(2,\alpha,v)$  with blocks $B_1,\dots,B_b\subset \mathcal{V}=[v]$, an FR code $\mathcal{C}=\{V_1,\dots,V_n\}$  can be obtained by taking
$$V_i=\{j|i\in B_j\},$$
for $i=1,\dots,n$.  This gives an FR code with $\rho=\alpha$ and   $\theta=\frac{n(n-1)}{\rho(\rho-1) } $ for a DSS with parameters  
 $n=v$ and $d=\frac{n-1}{\rho-1}$  as  given by Prop~\ref{prop:Steiner}.
\end{construction}

%

 We refer to the codes obtained by  this construction as  \emph{Transpose} codes since 
the role of the blocks and points  are reversed. The blocks now correspond to  packets  and  the points to   the  storage nodes. Therefore, any two nodes have exactly  one packet in common. Therefore, we get the following lemma.

\begin{lemma}\label{lem:TransposeG}
The FR codes obtained by Construction~\ref{cons:transpose} are universally good.
\end{lemma}

To highlight the difference between these two constructions, we give an example in Figure~\ref{fig:Transpose} when they are both applied to the unique Steiner system S(2,3,9) \cite[p. 27]{HandbookComb}.  Construction~\ref{cons:direct} gives an FR code with $\rho=4$ for a DSS with $n=12$ and $d=3$, whereas Construction~\ref{cons:transpose} gives an FR code with $\rho=3$ for a DSS with $n=9$ and $d=4$.  Note that these two constructions will give the same FR code (up to relabeling) when applied to projective planes such as  the Fano plane of Fig.~\ref{fig:FanoCode}(b).

%

The previous two constructions assume the existence of the Steiner system with the desired parameters, which is not always true.  However, Steiner systems $S(2,\alpha,v)$ are known to exist for  small values of $\alpha$, namely $\alpha=2 ,\dots, 5$, and for any $v$ whenever  the integrality conditions given by Th.~\ref{th:Wilson} are satisfied. This result in conjunction with Construction~\ref{cons:transpose} gives the necessary and sufficient conditions for the existence of Transpose codes with low repetition degree.

\begin{corollary}
Transpose codes  with repetition degree $\rho=2,\dots,5$ exist if and only if
$n-1\equiv 0 \bmod{\rho-1}$ and $n(n-1) \equiv 0 \bmod{\rho(\rho-1)}$.
\end{corollary}

The previous corollary implies that for the important practical case of systems with repetition degree $\rho=3$, universally good  FR codes  with repair degree $d=\frac{n-1}{2}$ can be obtained by Construction~\ref{cons:transpose} using Steiner systems $S(2,3,n)$ which exist for all $n\equiv 1,3 \mod{6}$. Steiner systems with $\alpha=3$, known in the literature as Steiner triple systems, are historically the most investigated systems and explicit constructions, such as Bose and Skolem constructions,   exist for all feasible values of $n$ \cite{DesignTh}.


\section{Capacity under Exact Uncoded Repair} \label{sec:FRCapacity}

The universally good FR codes constructed in the previous sections are guaranteed to have a rate greater or equal to the capacity $C_{MBR}$ of the system under random access and functional repair. However,  there exist cases where FR codes can achieve a storage capacity that exceeds   $C_{MBR}$.  
For instance, consider the FR code $\mathcal{C}=\{ \{1,2,3\}, \{4,5,6\},  \{7,8,9\}, \{1,4,7\}, \{2,5,8\}, \{3,6,9\}       \}$ 
of repetition degree 2 for the $(6,3,3)$ DSS  is depicted in Fig.~\ref{fig:Grid}.  It can be checked any user contacting $3$ nodes observes at least $7$ distinct packets. Therefore, this code has a rate $R_{\mathcal{C}}(3)=7> C_{MBR}=6$.

We refer to  the maximum file size that a DSS with parameters $ (n,k,d)$ can store under exact and uncoded repair  as its \emph{Fractional Repetition (FR) capacity} $C_{FR}$ defined as follows:
\begin{definition}[Fractional Repetition Capacity]
The Fractional Repetition (FR) capacity, denoted by $C_{FR} (k,\rho)$ of a distributed storage system  with parameters $ (n,k,d)$ is defined, for all $\rho$ satisfying $nd\equiv 0 \bmod{\rho}$, as 
$$C_{FR}(k,\rho):=\max_{\mathcal{C}}R_{\mathcal{C}}(k),$$
where $\mathcal{C}$ is any FR code with repetition degree $\rho$ for an  $(n,k,d)$ DSS. 
\end{definition}

The condition on $\rho$ in the definition above is needed by Prop.~\ref{prop:Steiner} to guarantee the existence of an FR code $\mathcal{C}$. Note that this notion of capacity assumes that a packet is an atomic unit of information that cannot be divided, which is usually true in real applications such the one in  \cite{GFS} where packets are of size 64 MB.

The code constructions of the previous sections imply lower bounds on the FR capacity. Next, we present two  upper bounds on  $C_{FR}$. The first is based on an averaging argument  and is presented in  Lemma~\ref{lem:average}. 

\begin{figure}[t]
\begin{center}
\includegraphics[width=1\columnwidth]{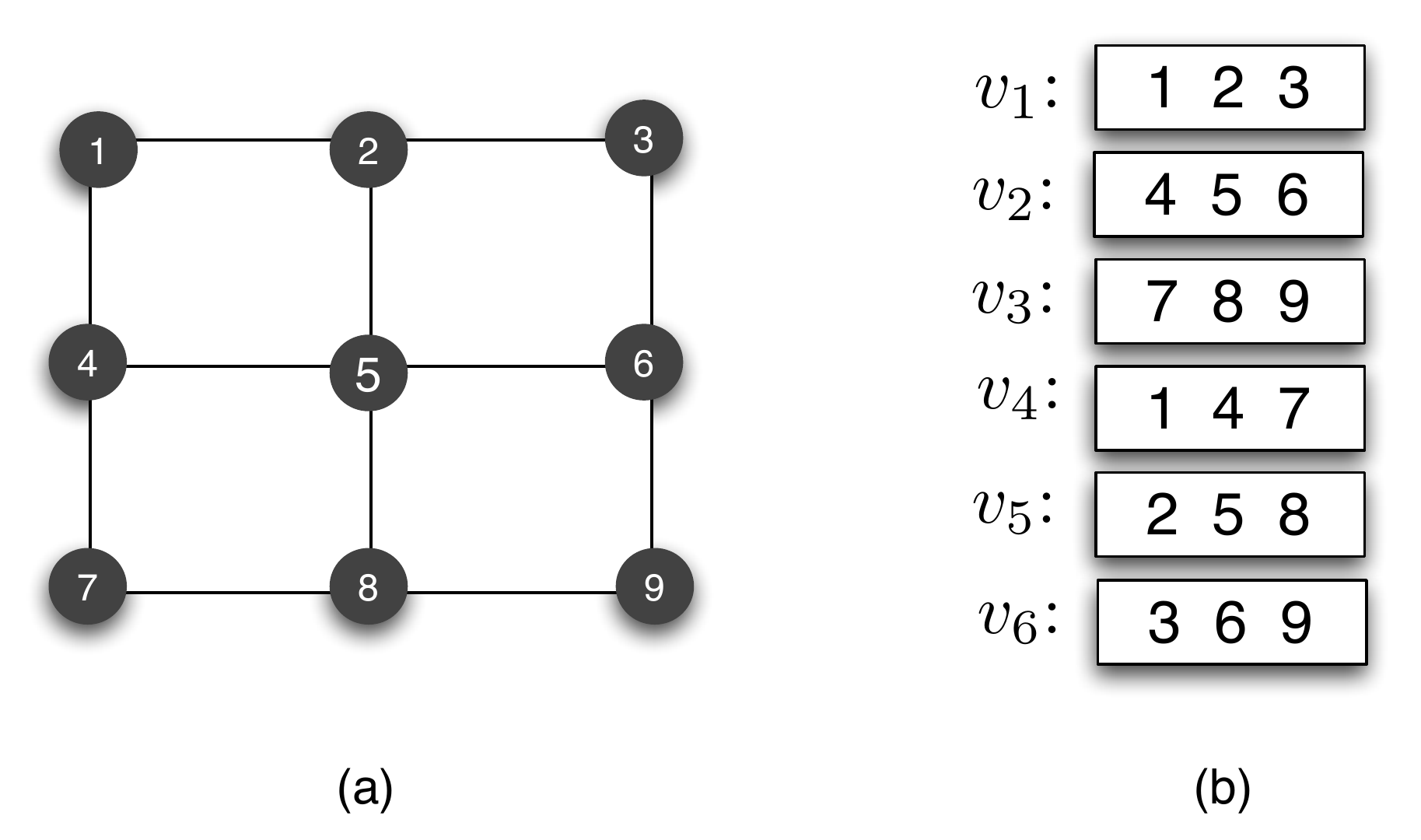}
 \caption{(a)  A $3\times 3$ grid  of $9$ points and $6$ lines. (b) The corresponding FR code achieving a storage capacity exceeding $C_{MBR}$.}\label{fig:Grid}
 \end{center}
\end{figure}

\begin{lemma} \label{lem:average}
For a  DSS with parameters $ (n,k,d)$,
$$C_{FR}(k,\rho)\leq \left \lfloor\frac{nd}{\rho} \left(1-\frac{\binom{n-\rho}{k}}{\binom{n}{k}}\right)\right \rfloor.$$
\end{lemma}
\begin{proof}
Let ${\cal C}=\{V_1,\dots, V_n\}$ be an FR code with repetition degree $\rho$, where $V_i\subset [\theta], |V_i|=d$ and $\theta=\frac{nd}{\rho}$ as given by Prop.~\ref{prop:theta}.

 Define the set ${\cal U}$ as 
$$\mathcal{U}:=\{U_I=\cup_{i\in I} V_i: I\subset [n], |I|=k\}.$$
The set $U_I$ represents the set of packets observed by a user contacting the nodes in the DSS indexed by the elements in $I$. We want to show that the term on the right in the inequality is the average cardinality of the sets in $\mathcal{U}$ under uniform distribution.  We denote this average by $\overline{U}$. To find $\overline{U}$, we count the following quantity $\sum_{U_I\in \mathcal{U}} |U_I|$ in two ways. 

First, we have by definition
$$\sum_{U_I\in \mathcal{U}} |U_I|=\binom{n}{k}\overline{U}.$$
But, each element in $[\theta]$ belongs to exactly $\binom{n}{k}-\binom{n-\rho}{k}$ sets in $\mathcal{U}$. Therefore, 
$$\sum_{U_I\in \mathcal{U}} |U_I|=\theta\left(\binom{n}{k}-\binom{n-\rho}{k} \right).$$
The upper bounds follows then from the fact that there must be in $\mathcal{U}$ at least one set of cardinality less that the average.
\end{proof} 

For instance, for the DSS $(7,3,3)$, Lem.~\ref{lem:average} implies that $R(3,3)\leq \lfloor 6.2 \rfloor =6$. Therefore, the FR code of Example~\ref{ex:Fano} is optimal and $C_{FR}(3,3)=6$.  However, the above upper bound has the disadvantage of  becoming loose for large values of $n$ and $k$ since the FR capacity is by definition a worst case measure. 

We also give a second bound  on  the FR capacity of a DSS which is defined using a recursive function and is tighter than the previous one.  

\begin{lemma}
For a DSS $(n,k,d)$, the FR capacity is upper bounded by the function  $g(k)$, \emph{i.e.},  
$C_{FR}(k,\rho)\leq g(k),$
where $g(k)$ is  defined recursively as 
\begin{align}
g(1)&=d,\\
g(k+1)&=g(k)+d-\left \lceil  \frac{\rho g(k)-kd}{n-k} \right \rceil.
\end{align}
\end{lemma}
%
%
%

The proof for this lemma is omitted  due to space restrictions and can be found in \cite{webfullversion}.

\section{Conclusion and Open Problems}\label{sec:Conclusion}

We proposed  a new class of Exact Minimum-Bandwidth Regenerating (MBR) codes for distributed storage systems characterized by a low complexity \emph{uncoded} repair process.  The main component of our  construction is  a new code that we call Fractional Repetition (FR) code. An FR code with repetition degree $\rho$ guarantees uncoded repair for up  to $\rho-1$ failures. It  consists of splitting the data on each node into multiple packets and storing  $\rho$ replicas of each on distinct nodes in the system.  An additional outer MDS code guarantees that a user contacting a sufficient number of storage nodes will  be able to retrieve the stored file. 
%

 For single node failures, \emph{i.e.},  $\rho=2$, we presented a construction of FR codes based on \emph{regular graphs} for all feasible system parameters. For the multiple failures case, \emph{i.e.}, $\rho>2$,  we presented two code constructions   based on \emph{Steiner systems}. Of particular importance are  the constructed \emph{Transpose} codes where the nodes contacted for repair are just a fraction of the surviving ones. All the obtained codes are guaranteed to achieve the storage capacity  under random-access repair. The adopted table-based repair model motivates a new concept of  capacity for distributed storage systems,  referred  to as of Fractional Repetition (FR) capacity,  which we studied and derived corresponding bounds.
 
 This work constitutes the first step in the study of Fractional Repetition codes and  many important questions remain open. For instance, it is not known whether  FR codes with $\rho>2$ exist for system parameters not covered by our constructions. Moreover,  a general expression of the FR capacity is still an open problem, as well as codes that can achieve it.

\bibliographystyle{ieeetr}
\bibliography{DSS}

\end{document}